\documentclass[pdflatex,sn-mathphys]{sn-jnl}

\jyear{2021}%

\theoremstyle{thmstyleone}%
\newtheorem{theorem}{Theorem}
\theoremstyle{thmstyletwo}%

\theoremstyle{thmstylethree}%
\newtheorem{definition}{Definition}%

\begin{document}

\title[Article Title]{Efficient Quantum Secret Sharing Scheme Based On Monotone Span Program}

\author[1]{\fnm{Shuangshuang} \sur{Luo}}

\author*[2]{\fnm{Zhihui} \sur{Li}}\email{lizhihui@snnu.cn;Tel:+86-130-3298-9886}

\author[3]{\fnm{Depeng} \sur{Meng}}

\author[4]{\fnm{Jiansheng} \sur{Guo}}

\affil{\orgdiv{School of Mathematics and Statistics}, \orgname{Shaanxi Normal University}, \orgaddress{ \city{Xi'an}, \country{China}}}

\affil*{\orgdiv{School of Mathematics and Statistics}, \orgname{Shaanxi Normal University}, \orgaddress{ \city{Xi'an}, \postcode{710119}, \country{China}}}

\affil{\orgdiv{School of Mathematics and Statistics}, \orgname{Shaanxi Normal University}, \orgaddress{ \city{Xi'an}, \country{China}}}

\affil{\orgdiv{School of Mathematics and Statistics}, \orgname{Shaanxi Normal University}, \orgaddress{ \city{Xi'an}, \country{China}}}


\abstract{How to efficiently share secrets among multiple participants is a very important problem in key management. In this paper, we propose a multi-secret sharing scheme based on the GHZ state. First, the distributor uses monotone span program to encode the secrets and generate the corresponding secret shares to send to the participants. Then, each participant uses the generalized Pauli operator to embed its own secret share into the transmitted particle. The participant who wants to get the secrets can get multiple secrets at the same time by performing a GHZ-state joint measurement. Futhermore, the scheme is based on a monotone span program, and its access structure is more general than the access structure $(t,n)$ threshold. Compared with other schemes, our proposed scheme is more efficient, less computational cost.}

\keywords{Access structure, Monotone Span Program, GHZ state, Quantum secret sharing}


\maketitle

\section{Introduction}\label{sec1}

Secret sharing plays an important role in the field of cryptography. Its main idea is to divide the secret into several shares and distribute them to participants. Only the participants in the authorized set can cooperate to reconstruct the secret. Any one or more unauthorized participants cannot reconstruct the secret, which makes secret sharing significant in information security, because multiple participants hold secret shares, it can spread risks and prevent secrets from being destroyed or changed. In 1979, Shamir \cite{ref-ur1} and Blakley \cite{ref-ur2} proposed the first classical threshold secret sharing (TSS) based on Lagrange interpolation, and applied it to various fields \cite{ref-ur3,ref-ur4,ref-ur5,ref-ur6,ref-ur7,ref-ur8,ref-ur9}. However, the security of classical secret sharing depends on the assumption of computational complexity, which cannot provide unconditional secure communication. To solve these problems, in 1984, C.H.Bennett and G.Brassard jointly proposed the famous BB84 protocol \cite{ref-ur10}, which marked the birth of quantum cryptography. In 1999, Hillery et al.\cite{ref-ur11}  proposed the first quantum secret sharing (QSS) scheme based on the entanglement properties of GHZ (Greenberger Horne Zeilinger) state. Since then, more and more quantum secret sharing schemes \cite{ref-ur12,ref-ur13,ref-ur14,ref-ur15,ref-ur16,ref-ur17,ref-ur18,ref-ur19,ref-ur20,ref-ur21} have been proposed. Compared with classical secret sharing, quantum secret sharing shares secrets through quantum operations, so, the security of the scheme is based on quantum theory, such as the non-cloning theorem of quantum mechanics, Heisenberg uncertainty principle and the principle of indistinguish ability of non-orthogonal quantum states. Quantum secret sharing has attracted extensive attention of many cryptologists and become a research hotspot of scholars. It has been rapidly developed and applied in a short time. For example, In 2019, Chen et al.\cite{ref-ur12} proposed a flexible evaluator two-dimensional quantum homomorphic encryption scheme based on $(k,n)$ threshold quantum state sharing, In 2022, Wang et al.\cite{ref-ur14} constructed a threshold QSS scheme by using rotation operation. At the same time, some quantum secret sharing schemes in high-dimensional space have been proposed \cite{ref-ur15,ref-ur16,ref-ur17,ref-ur18,ref-ur19,ref-ur21}. However, there are problems with some of these schemes, for example, in 2018, Song et al. \cite{ref-ur13} proposed a QSS scheme with $(t,n)$ threshold. In this scheme, any t participants perform unitary operation based on the Lagrange interpolation, and then the participants can recover the secret by performing quantum inverse Fourier transform. However, Kao et al. \cite{ref-ur22} found that the scheme in \cite{ref-ur13} has a defect, that is, the participant cannot successfully recover the secret. In 2019, Mashhadi \cite{ref-ur19} inspired by the quantum Fourier transform (QFT), proposed a general QSS scheme based on the monotone span program. However, this scheme cannot conduct eavesdropping check during the process of the particle transmission, which leads to weak security of the scheme. In 2021, Li et al. \cite{ref-ur23} proposed a general QSS scheme based on two qudit Bell state similar to Mashhadi \cite{ref-ur20}'s scheme, inserting decoy particles to improve security. In literature \cite{ref-ur24}, Lu et al. used unitary operation to construct a quantum secret sharing scheme, the disadvantage of this scheme is that it is impossible to verify whether the recovered secret is the original secret. On the other hand, in the quantum secret sharing schemes \cite{ref-ur12,ref-ur13,ref-ur14,ref-ur15,ref-ur16,ref-ur17,ref-ur18,ref-ur19,ref-ur20,ref-ur21} mentioned above, threshold quantum secret sharing schemes accounts for a large proportion. For example, $(n,n)$ threshold QSS schemes \cite{ref-ur15,ref-ur19} and $(t,n)$ threshold QSS schemes \cite{ref-ur12,ref-ur13,ref-ur14,ref-ur16,ref-ur17,ref-ur18}. In the $(t,n)$ threshold QSS scheme, the size of each authorized set is at leas $t(t \leq n)$, that is, the secret can be recovered only when at least t participants participate in cooperation, and the secret cannot be recovered if less than t participants. However, considering some practical situations, such as some cashiers and managers sharing the keys of bank safe, the rights of cashiers and managers are different. At this time, the $(t,n)$ threshold QSS is no longer suitable, so the access structure is introduced to make QSS more general. The access structure is a set of subsets of all participants authorized to reconstruct the secret. In the general QSS scheme, the secret is shared among participants, and only the participants in the authorized set can recover the secret \cite{ref-ur20,ref-ur21}, each authorized set can have a different number of participants, so the general QSS scheme is an extension of the $(t,n)$ threshold QSS scheme. Compared with the $(t,n)$ threshold QSS scheme, the general QSS scheme is more practical. \\ 
\indent For the construction of QSS schemes, it is more and more important to verify the correctness of the recovered secrets and improving the efficiency of the schemes while ensuring the security of schemes. In 2013, Yang et al. \cite{ref-ur15} used multiparty computation to construct a multi-dimensional QSS scheme based on quantum Fourier transform (QFT), but this scheme could not verify the correctness of the recovered secret. Although the literature \cite{ref-ur20} uses hash function to check the correctness of the recovered secret, it uses twice hash. In 2022, Wu et al. \cite{ref-ur21} proposed a secret sharing scheme based on d dimensional three-qudit GHZ state with adversary structure. The GHZ state is transmitted between the distributor and the participants. The participants perform unitary transformation on the GHZ state in turn. The last participant performs a GHZ-state joint measurement and publishes the result. The first quantum state is used as the sharing secret information, and the last two quantum states are used as the verification information. Although the scheme verifies the secret, its efficiency is relatively low, that is, two of the three-qudit GHZ state are used as verification information instead of sharing secret information, which wastes quantum resources. Therefore, to solve this problem, we construct a new quantum secret sharing scheme based on three-qudit GHZ state. Our scheme can guarantee the correctness of the recovered secret through once hash, and has low computational cost. In our scheme, all three particles of three-qudit GHZ state are used to share secret information, which improves the efficiency.

The structure of this paper is as follows: In Section 2, we briefly introduce some preliminary knowledge required by the scheme. In Section 3, we describe a secret sharing scheme in detail, and in Section 4, we give a simple example. In Section 5, we discuss the correctness, verifiability and security of the scheme. In Section 6, we analyze and compare existing schemes to illustrate the advantages of the scheme. Finally, In Section 7, we give conclusions.

\section{Preliminaries}\label{sec2}

In this section, we will briefly introduce some preliminary knowledge required by the scheme, including unitary operation, monotone span program and hash function.

\subsection{Three-qudit GHZ State and Unitary Operation}\label{subsec2}

Let $d$ be an odd prime and $F_d$ be a finite field. In the $d$-dimensional Hilbert space, the generalized $n$-qudit GHZ state is described as follows.
\begin{equation}
\mid\varPsi_{{u_1},{u_2},\dots,{u_n}}\rangle=\frac{1}{\sqrt{d}}\sum_{j=0}^{d-1} \omega^{j{u_1}}\mid{j,j+{u_2,\dots,j+u_n}}\rangle,\label{eq1}
\end{equation}
where $\omega=e^\frac{2\pi i}{d},u_i\in\lbrace0,1,\dots,d-1\rbrace, i=1,2,\dots,n$. The symbol "+" is an addition in the finite field $F_d$.

Specially, the generalized three-qudit GHZ state in the $d$-dimensional Hilbert space is described as follows.
\begin{equation}
\mid\varPsi_{{u_1},{u_2},{u_3}}\rangle=\frac{1}{\sqrt{d}}\sum_{j=0}^{d-1} \omega^{j{u_1}}\mid{j,j+{u_2,j+u_3}}\rangle,\label{eq2}
\end{equation}
Further, the generalized Pauli operator can be described in the following form:
\begin{equation}
U_{{\alpha},{\beta},{0}}=\sum_{j=0}^{d-1} \omega^{j{\alpha}}\mid{j}\rangle\langle{j+\beta}\mid,\label{eq3}
\end{equation}
where $\alpha,\beta\in\lbrace0,1,\dots,d-1\rbrace$. The symbol "+" is an addition in the finite field $F_d$. And the subscript "0" of denotes the initial value of the third position.

\subsection{Monotone Span Program}\label{subsec2}

Assume that $P=\lbrace{{P_1},{P_2},\dots,{P_n}}\rbrace$ is the set of all participants in a secret sharing scheme, $\varGamma$ is a subset of $2^p$.  If the participants contained in the element of $\varGamma$ can restore secrets, we call $\varGamma$ is the access structure in $P$,  the element in $\varGamma $ that is a authorized subset. In this article, let $\varGamma=\lbrace{A_1,A_2,\dots,A_r}\rbrace$, where $A_i\in \varGamma,i=1,2,\dots,r$ is authorized set.
A monotonic access structure in $P$ is a collection of a subset of the set $P$ and meet the monotone increasing properties: $B\in\varGamma$, when $A_i\in \varGamma $, $A_i\subseteq B\subseteq P, i=1,2,\dots,r $. A non-access structure, denoted as $\Delta $, that is a collection of a subset of the set $P$ and meet the monotone decreasing properties: $A\in\Delta$, when $B\in \Delta $, $A\subseteq B\subseteq P $. Access structure and non-access structure satisfy the relationship: $\Delta=\varGamma^c $.

\begin{definition}\label{Def1}
Monotone span program is a four-tuple $(F_d,M,f,\xi)$, $M$ is a $k\times l$ matrix on the finite field $F_d$. Let row $i$ of matrix $M$ is $\alpha_i$, and do the mapping $f$, where $f:\lbrace{{\alpha_1},{\alpha_2},\dots,{\alpha_k}}\rbrace\to P$ is a surjection which used to allocate rows of matrix $M$ to each participant; $\xi=(1,0,\dots,0)^T\in F_d$ is called target vector.
\end{definition} 

\begin{definition} \label{Def2}
For access structure $\varGamma=\lbrace{A_1,A_2,\dots,A_r}\rbrace$, if the following conditions are met: \\
(1) For any $A_i\in\varGamma$, there exists a vector $\lambda_{A_i}\in F_d^{m_i} $ such that $M_{A_i}^T\lambda_{A_i}=\xi$. \\
(2) For any $A_i\in \varGamma$, there exists a vector $h=(1,h_2,h_3,\dots,h_l)\in F_d^{m_i}$ such that $M_{A_i}h\in F_d^{m_i}$, which is called the sweeping vector for $A_i$.\\
it needs to be said $m_i$ in definition 2 is the number of players in $A_i$, and $M_{A_i}$ denotes the rows $j$ of $M$ such that $f(\alpha_j)\in A_i$. Then the four-tuple $(F_d,M,f,\xi)$ is called a monotone span program.
\end{definition}
according to the monotone span program $(F_d,M,f,\xi)$, we can construct line secret sharing on access structure $\varGamma$. Line secret sharing includes distribution phase and reconstruction phase. In the distribution phase, Alice constructs a vector $\rho$, where the first component of the vector is a secret, computes $M_j\rho$, and sends the value to the participant $P_j$ by a secure channel \cite{ref-ur32,ref-ur33}, where $M_j$ represents the rows of the matrix $M$ owned by the participant$P_j$. In the distribution phase, for the authorized set $A_i$, exists vector $\lambda_{A_i}$, which makes the participants in $A_i$ can recover the secret $s$ as follows.
\begin{equation}
s_{A_i}\lambda_{A_i}=(M_{A_i}\rho)^T\lambda=\rho^T(M_{A_i}^T\lambda_{A_i})=\rho^T\xi=s
\label{eq4}
\end{equation}

\subsection{Hash Function}\label{subsec3}

Hash function has important applications in cryptography\cite{ref-ur25,ref-ur26,ref-ur27,ref-ur28,ref-ur29,ref-ur30,ref-ur31}. Many schemes use hash function $H(\cdot)$ to verify the correctness of the recovered secret. In this article, hash function $H(\cdot):F_d \to F_d$ and it satisfied the following properties:\\
(1)Computability: known $x \in F_d $, easy to calculate $H(x)$;\\
(2)Unidirectivity: given a value $y\in F_d$, it is computationally infeasible to find $x$ satisfying $H(x)=y$;\\
(3)Weak collision resistance: given a value $x\in F_d$, find out $y\in F_d(y\ne x)$, so that $H(x)=H(y)$ is computationally infeasible;\\
(4)Strong collision resistance: find any two different input values $x$, $x^\prime$, so that $H(x)=H(x^\prime)$ is not feasible in calculation.

\section{Our Scheme}\label{sec3}

In this section, we will construct a verifiable quantum secret sharing scheme, which can transmit multiple secrets by one time. This scheme includes Alice and $n$ participants $P_1$, $P_2$,...,$P_n$. In this scheme, we assume that $H(\cdot)$ is a public hash function that satisfies 2.3 in last section. The access structure $\varGamma=\lbrace{A_1,A_2,\dots,A_r}\rbrace$, where $A_i\in \varGamma,i=1,2,\dots,r$, is authorized set. In this article we denote the authorized set as $A_i=\lbrace{P_1^{(i)},P_2^{(i)},\dots,P_{m_i}^{(i)}}\rbrace,1\leq m_i\leq n$. About $\varGamma$, Alice constructs a monotone span program $(F_d,M,f,\xi)$ on $\varGamma$  as described in Section 2.2 of this article, $M$ is a $n\times l$ matrix on the finite field $F_d$; mapping $f$ satisfies $f(\alpha_i)=P_i$. According to the monotone span program $(F_d,M,f,\xi)$, we can construct line secret sharing on access structure $\varGamma$. Line secret sharing includes distribution phase and reconstruction phase. In the distribution phase, Alice constructs a vector $\rho$, where the first component of the vector is a secret, computes $M_j\rho$, and sends the value to the participant $P_j$ by a secure channel \cite{ref-ur32,ref-ur33}, where $M_j$ represents the rows of the matrix $M$ owned by the participant$P_j$. In the distribution phase, according to (4), the participants in the authorized set $A_i$ can recover the secret $s$.

\subsection{Preparation Phase}\label{subsec1}

Suppose Alice wants to share classical information $(s_1,s_2,s_3), s_i\in F_d, i=1,2,3$, then she needs to do the following operations:

\indent\textbf{3.1.1} Alice prepares randomly a three-qudit GHZ state in the $d$-dimensional Hilbert space:

\[\mid\varPsi_{{u_1},{u_2},{u_3}}\rangle=\frac{1}{\sqrt{d}}\sum_{j=0}^{d-1} \omega^{j{u_1}}\mid{j,j+{u_2,j+u_3}}\rangle,\qedhere\]
where $\omega=e^\frac{2\pi i}{d},u_i\in\lbrace0,1,\dots,d-1\rbrace, i=1,2,\dots,n$. The symbol "+" is an addition in the finite field $F_d$. Three particles of $\mid\varPsi_{{u_1},{u_2},{u_3}}\rangle$ are called as $Q_1,Q_2,Q_3$.

\indent\textbf{3.1.2} Alice prepares randomly some decoy particles in $Z$ basis and $X$ basis.\\
The $Z$ basis and $X$ basis have the following forms:
\begin{align}
 Z &=\lbrace\mid i\rangle\rbrace,\notag\\
 X &=\lbrace\mid J_j\rangle\rbrace,\notag\
\end{align}
where $\mid J_j\rangle=\frac{1}{\sqrt{d}}\sum_{j=0}^{d-1}\omega^{kj}\mid{k}\rangle\ ,i,j \in \lbrace0,1,\dots,d-1\rbrace $.

\indent\textbf{3.1.3} According to the secrets $s_1, s_2$, Alice constructs vectors
\begin{align}
 \rho_1 &=(s_1,a_2^{(1)},a_3^{(1)},\dots,a_{m_i}^{(1)})^T,\notag\\
 \rho_2 &=(s_2,a_2^{(2)},a_3^{(2)},\dots,a_{m_i}^{(2)})^T,\notag\
\end{align}
where $a_2^{(1)},a_3^{(1)},\dots,a_{m_i}^{(2)}\in F_d $.

\indent\textbf{3.1.4} Alice according to the monotone span program, for the authorized set $A_i$ in the access structure $\varGamma$, $M_{A_i}$ can be obtained by mapping $f$, calculates
\begin{align}
 \ S^{(i)} &=M_{A_i}\rho_1=(s_1^{(i)},s_2^{(i)},\dots,s_{m_i}^{(i)})^T,\notag\\
 \ T^{(i)} &=M_{A_i}\rho_2=(t_1^{(i)},t_2^{(i)},\dots,t_{m_i}^{(i)})^T.\notag\
\end{align}
And the unique solution $\lambda_{A_i}$ can be obtained according to the $M_{A_i}$ and $\xi$.

\indent\textbf{3.1.5} Alice uses the public hash function $H(\cdot)$, computes and publishes $H(s_1),H(s_2),H(s_3)$.

\subsection{Distribution Phase}\label{subsec2}

\textbf{3.2.1} Alice sends the corresponding $j$ th component of the vectors $S^{(i)},T^{(i)}$ and $\lambda_{A_i}$ as secret shares to the participants $P_j^{(i)}$ in the authorized set $A_i$ through the secure quantum channel. $j=1,2,\dots,m_i$.

\indent\textbf{3.2.2} Alice inserts the particle $Q_1$ into the decoy particles to form the particle sequences $C_1$. And records the inserted position and initial state of the decoy particles. Then Alice sends the particle sequences $C_1$ to the participant $P_1^{(i)}$ in the authorized $A_i$.

\subsection{Recovery Phase}\label{subsec3}

In our scheme, the secrets recovery phase can be divided into several steps as follows.

\indent\textbf{3.3.1} After ensuring the participant $P_1^{(i)}$ receives these particle sequences $C_1$, Alice publicly informs participant $P_1^{(i)}$ of the basis with the decoy particles, and then asks the participant $P_1^{(i)}$ to select the $Z$ basis or the $X$ basis to measure these particles according to the basis of the decoy particle. Next, $P_1^{(i)}$ publishes the measurement results. Alice calculates the error ratio by comparing the measurement results with initial state of decoy particles. If the value exceeds the predetermined threshold (usually not more than 11\% \cite{ref-ur34}), Alice asks the participant $P_1$ to stop this operation and start a new round. Otherwise, continue to execute these steps. 

\indent\textbf{3.3.2} Participant $P_1^{(i)}$ prepares a random number $q_1$, and performs unitary transformation
$U_{{q_1},{q_1},{0}}=\sum_{j=0}^{d-1} \omega^{{q_1}j}\mid{j}\rangle\langle{j+{q_1}}\mid\notag$
on particle $Q_1$, then the GHZ state 
$\mid\varPsi_{{u_1},{u_2},{u_3}}\rangle$
will become $\mid\varPsi_{{u_1}+{q_1},{u_2}+{q_1},{u_3}+{q_1}}\rangle$. Then, participant $P_1^{(i)}$ prepares some decoy particles on the $Z$ basis and $X$ basis, and inserts the particle $Q_1$ into the decoy particles to form the particle sequences $C_2$, then transmits these particle sequences $C_2$ to the next participant $P_2^{(i)}$. Similarly, participant $P_1^{(i)}$ can use these decoy particles to check the security of the channel between $P_1^{(i)}$ and $P_2^{(i)}$. Participant $P_2^{(i)}$ performs unitary transformation
\begin{equation}
U_{{\lambda_{A_i}^{(2)}s_2^{(i)}},{\lambda_{A_i}^{(2)}t_2^{(i)}},{0}}=\sum_{j=0}^{d-1} \omega^{{\lambda_{A_i}^{(2)}s_2^{(i)}}j}\mid{j}\rangle\langle{j+{\lambda_{A_i}^{(2)}t_2^{(i)}}}\mid\notag
\end{equation}
on particle $Q_1$, then the GHZ state $\mid\varPsi_{{u_1}+{q_1},{u_2}+{q_1},{u_3}+{q_1}}\rangle\ $ will become 
\[\mid\varPsi_{{u_1}+{q_1}+{\lambda_{A_i}^{(2)}s_2^{(i)}},{u_2}+{q_1}+{\lambda_{A_i}^{(2)}t_2^{(i)}},{u_3}+{q_1+{\lambda_{A_i}^{(2)}t_2^{(i)}}}}\rangle\]  
Where $\lambda_{A_i}^{(2)}$ is the private value of participant $P_2^{(i)}$, $s_2^{(i)}$ and $t_2^{(i)}$ are the corresponding 2 th component of the vectors $S^{(i)},T^{(i)}$. Repeat this operation until the participant $P_{m_i}^{(i)}$ to get the GHZ state
\[\mid\varPsi_{{u_1}+{q_1}+\sum_{j=2}^{m_i}{\lambda_{A_i}^{(j)}s_j^{(i)}},{u_2}+{q_1}+\sum_{j=2}^{m_i}{\lambda_{A_i}^{(j)}t_j^{(i)}},{u_3}+{q_1}+\sum_{j=2}^{m_i}{\lambda_{A_i}^{(j)}t_j^{(i)}}}\rangle\]  

\indent\textbf{3.3.3} After all above operations are completed, Alice performs a unitary transformation $U_{{d-u_1},{d-u_2},{0}}$ on the particle $Q_2$, then Alice performs a unitary transformation $U_{{0},{d+u_3+s_2-s_3},{0}}$ on the particle $Q_3$, then the GHZ state 
\[\mid\varPsi_{{u_1}+{q_1}+\sum_{j=2}^{m_i}{\lambda_{A_i}^{(j)}s_j^{(i)}},{u_2}+{q_1}+\sum_{j=2}^{m_i}{\lambda_{A_i}^{(j)}t_j^{(i)}},{u_3}+{q_1}+\sum_{j=2}^{m_i}{\lambda_{A_i}^{(j)}t_j^{(i)}}}\rangle\] 
will become
$\mid\varPsi_{{q_1}+\sum_{j=2}^{m_i}{\lambda_{A_i}^{(j)}s_j^{(i)}},{q_1}+\sum_{j=2}^{m_i}{\lambda_{A_i}^{(j)}t_j^{(i)}},{q_1}+\sum_{j=2}^{m_i}{\lambda_{A_i}^{(j)}t_j^{(i)}-s_2+s_3}}\rangle $. Last, Alice sends the particle $Q_2$ and $Q_3$ to the participant $P_1^{(i)}$. (this step also requires the use of decoy particles for security checks)

\indent\textbf{3.3.4} Participant $P_{m_i}^{(i)}$ sends the particle $Q_1$ to the participant $P_1^{(i)}$ at the same. After participant $P_1^{(i)}$ receives the particle $Q_1$, participant $P_1^{(i)}$ performs a unitary transformation
\begin{equation}
U_{{\lambda_{A_i}^{(1)}s_1^{(i)}-q_1},{\lambda_{A_i}^{(1)}t_1^{(i)}-q_1},{0}}=\sum_{j=0}^{d-1} \omega^{{q_1}j}\mid{j}\rangle\langle{j+{\lambda_{A_i}^{(1)}t_1^{(i)}-q_1}}\mid\notag
\end{equation}
on the particle $Q_1$, the GHZ state
\[\mid\varPsi_{{q_1}+\sum_{j=2}^{m_i}{\lambda_{A_i}^{(j)}s_j^{(i)}},{q_1}+\sum_{j=2}^{m_i}{\lambda_{A_i}^{(j)}t_j^{(i)}},{q_1}+\sum_{j=2}^{m_i}{\lambda_{A_i}^{(j)}t_j^{(i)}-s_2+s_3}}\rangle\] 
will become
$\mid\varPsi_{\sum_{j=1}^{m_i}{\lambda_{A_i}^{(j)}s_j^{(i)}},\sum_{j=1}^{m_i}{\lambda_{A_i}^{(j)}t_j^{(i)}},\sum_{j=1}^{m_i}{\lambda_{A_i}^{(j)}t_j^{(i)}-s_2+s_3}}\rangle $, that is $\mid\varPsi_{{s_1},{s_2},{s_3}}\rangle$. Eventually, participant $P_1^{(i)}$ performs GHZ-state joint measurement, where $\mid\varPsi_{{u_1^\prime},{u_2^\prime},{u_3}^\prime}\rangle$ is new state obtained by Alice and $m_i$ participants performing unitary transformation.

\indent\textbf{3.3.5} According to the GHZ-state joint measurement of $\mid\varPsi_{{u_1^\prime},{u_2^\prime},{u_3}^\prime}\rangle$, the participant $P_1^{(i)}$ can get the following three equations
\begin{align}
 \ s_1 &=u_1^\prime,\\
 \ s_2 &=u_2^\prime,\\
 \ s_3 &=u_3^\prime,
\end{align}
The participant $P_1^{(i)}$ can respectively verify whether $H(u_j^\prime)$ and $H(s_j)$ are equal can according to equations (5), (6) and (7), $j=1,2,3$. If the equation holds, participant $P_1^{(i)}$ can decide that the recovered secrets are correct; otherwise, it can judge that in the process of the secrets recovery, some participants are dishonest, the process can be stopped.\\
\indent Other participants $P_j^{(i)}$ in the authorized set $A_i$ can obtain secrets in the same way. The particle sequence transmission in this scheme is shown in Figure 1.

\begin{figure}[h]%
\centering
\includegraphics[width=\textwidth]{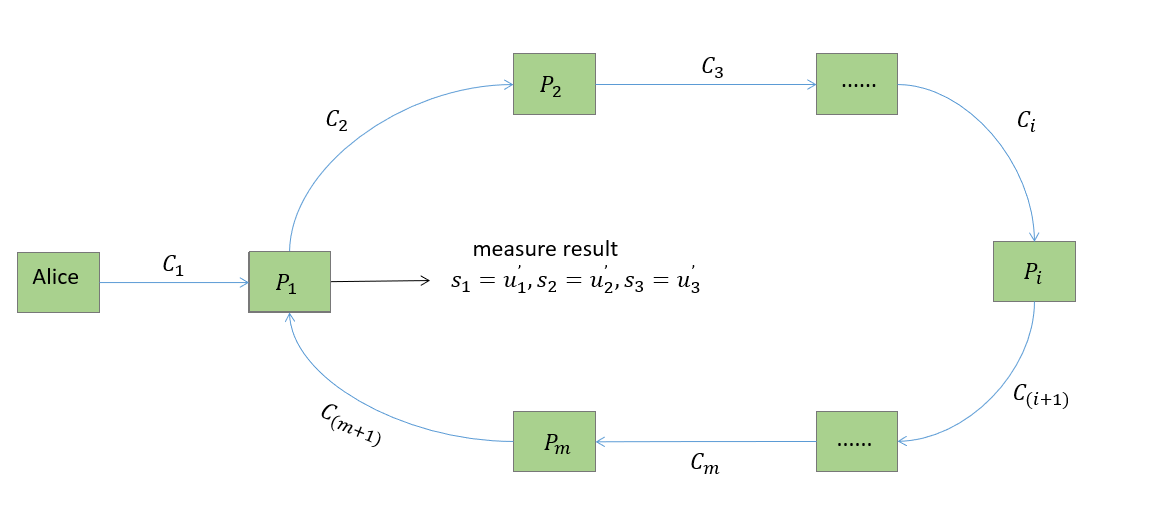}
\centering
\caption{The particle sequence transmission in this scheme}\label{fig1}   
\end{figure}

\section{Example}\label{sec4}
To better understand the above steps, here we showcase one example.

Assume the participant $P=\lbrace{P_1,P_2,P_3,P_4}\rbrace$, the access structure $\varGamma=\lbrace A_1=P,A_2=\lbrace{P_1,P_2,P_3}\rbrace,A_3=\lbrace{P_1,P_3,P_4}\rbrace \rbrace$. About $\varGamma$, Alice constructs a monotone span program $(F_7,M,f,\xi)$ on $\varGamma$, where $M=\begin{pmatrix}
    0 & 0 & 1 & 1 \\
    0 & 1 & 0 & 6 \\
    2 & 1 & 1 & 0 \\
    3 & 1 & 2 & 1
\end{pmatrix}$. Mapping $f$ satisfies $f(\alpha_i)=P_i,i\in\lbrace{1,2,3,4}\rbrace$.$\alpha_i$ is row $i$ of matrix $M$, do the mapping $f$. For the convenience of narration, we choose the authorized set $A_1$ in the access structure $\varGamma$ which is used to transmit the GHZ state, and $P_1$ recovers the final secrets, If Alice wants to share the classical information $(s_1,s_2,s_3)=(2,4,6)$, she must perform the following operations:  

\subsection{Preparation Phase}\label{subsec1}

\textbf{4.1.1} Alice prepares randomly a three-qudit GHZ state $\mid\varPsi_{6,1,3}\rangle$ in the 7-dimensional Hilbert space. Three particles of $\mid\varPsi_{6,1,3}\rangle$ are called as $Q_1,Q_2,Q_3$.

\indent\textbf{4.1.2} Alice prepares randomly some decoy particles in $Z$ basis and $X$ basis.
\indent\textbf{4.1.3} According to the secrets $s_1=2,s_2=4$, Alice constructs vectors
\begin{align}
 \rho_1 &=(2,1,0,5)^T,\notag\\
 \rho_2 &=(4,3,5,1)^T,\notag\
\end{align}

\indent\textbf{4.1.4} Alice according to the monotone span program, for the authorized set $A_1$ in the access structure $\varGamma$,$M_{A_1}=M$ can be obtained by mapping $f$, calculates
\begin{align}
 \ S^{(1)} &=M\rho_1=(s_1^{(1)},s_2^{(1)},s_3^{(1)},s_4^{(1)})^T=(5,3,5,5)^T,\notag\\
 \ T^{(1)} &=M\rho_2=(t_1^{(1)},t_2^{(1)},t_3^{(1)},t_4^{(1)})^T=(6,2,2,5)^T.\notag\
\end{align}
And the unique solution $\lambda_{A_1}=(3,3,4,0)^T$ can be obtained according to the $M_{A_1}$ and $\xi$.

\indent\textbf{4.1.5} Alice uses the public hash function $H(\cdot)$, computes and publishes $H(2),H(4),H(6)$.

\subsection{Distribution Phase}\label{subsec2}

\textbf{4.2.1} Alice sends the corresponding $j$ th component of the vectors $S^{(1)},T^{(1)}$ and $\lambda_{(A_1)}$ as secret shares to the participants $P_j^{(1)}$ in the authorized set $A_1$ through the secure quantum channel.$j=1,2,3,4$.

\indent\textbf{4.2.2} Alice inserts the particle $Q_1$ into the decoy particles to form the particle sequences $C_1$. And records the inserted position and initial state of the decoy particles. Then Alice sends the particle sequences $C_1$ to the participant $P_1^{(1)}$ in the authorized $A_1$.

\subsection{Recovery Phase}\label{subsec3}

\textbf{4.3.1} After ensuring the participant $P_1^{(1)}$ receives these particle sequences $C_1$, Alice publicly informs participant $P_1^{(1)}$ of the basis with the decoy particles, and then asks the participant $P_1^{(1)}$ to select the $Z$ basis or the $X$ basis to measure these particles according to the basis of the decoy particle. Next, $P_1^{(1)}$ publishes the measurement results. Alice calculates the error ratio by comparing the measurement results with initial state of decoy particles. If the value exceeds the predetermined threshold (usually not more than 11\% \cite{ref-ur34}), Alice asks the participant $P_1^{(1)}$ to stop this operation and start a new round. Otherwise, continue to execute these steps. 

\indent\textbf{4.3.2} Participant $P_1^{(1)}$ prepares a random number $q_1=5$, and performs unitary transformation $U_{{5},{5},{0}}=\sum_{j=0}^{d-1} \omega^{5j}\mid{j}\rangle\langle{j+5}\mid\notag $ on particle $Q_1$, then the GHZ state $\mid\varPsi_{6,1,3}\rangle$  will become $\mid\varPsi_{4,6,1}\rangle$. Then participant $P_1^{(1)}$ prepares some decoy particles on the $Z$ basis and $X$ basis, and inserts the particle $Q_1$ into the decoy particles to form the particle sequences $C_2$, then transmits these particle sequences $C_2$ to the next participant $P_2^{(1)}$. Similarly, participant $P_1^{(1)}$ can use these decoy particles to check the security of the channel between $P_1^{(1)}$ and $P_2^{(1)}$. Participant $P_2^{(1)}$ performs unitary transformation
$U_{2,6,0}=\sum_{j=0}^{6} \omega^{2j}\mid{j}\rangle\langle{j+6}\mid\notag$ on particle $Q_1$, then the GHZ state $\mid\varPsi_{4,6,1}\rangle$ will become 
$\mid\varPsi_{6,5,0}\rangle$,  $(\lambda_{A_1}^{(2)}s_2^{(1)}=3\times3=2(\bmod{7}), \lambda_{A_1}^{(2)}t_2^{(1)}=3\times2=6(\bmod{7}))$, where $\lambda_{A_1}^{(2)}$ is the private value of participant $P_2^{(1)}$, $s_2^{(1)}$ and $t_2^{(1)}$ are the corresponding 2 th component of the vectors $S^{(1)},T^{(1)}$. Repeat this operation, participant $P_3^{(1)}$ performs unitary transformation
$U_{6,1,0}=\sum_{j=0}^{6} \omega^{6j}\mid{j}\rangle\langle{j+1}\mid\notag$ on particle $Q_1$, then the GHZ state $\mid\varPsi_{6,5,0}\rangle$ will become 
$\mid\varPsi_{5,6,1}\rangle$,  $(\lambda_{A_1}^{(3)}s_3^{(1)}=4\times5=6(\bmod{7}), \lambda_{A_1}^{(3)}t_3^{(1)}=4\times2=1(\bmod{7}))$, where $\lambda_{A_1}^{(3)}$ is the private value of participant $P_3^{(1)}$, $s_3^{(1)}$ and $t_3^{(1)}$ are the corresponding 3 th component of the vectors $S^{(1)},T^{(1)}$. Participant $P_4^{(1)}$ performs unitary transformation $U_{0,0,0}=\sum_{j=0}^{6} \omega^{0j}\mid{j}\rangle\langle{j+0}\mid\notag$
on particle $Q_1$, then the GHZ state $\mid\varPsi_{5,6,1}\rangle$ will become 
$\mid\varPsi_{5,6,1}\rangle$,  $(\lambda_{A_1}^{(4)}s_4^{(1)}=0\times5=0(\bmod{7}), \lambda_{A_1}^{(4)}t_4^{(1)}=0\times5=0(\bmod{7}))$, where $\lambda_{A_1}^{(4)}$ is the private value of participant $P_4^{(1)}$, $s_4^{(1)}$ and $t_4^{(1)}$ are the corresponding 4 th component of the vectors $S^{(1)},T^{(1)}$.

\indent\textbf{4.3.3} After all above operations are completed, Alice performs a unitary transformation $U_{{7-6},{7+1},{0}}$ on the particle $Q_2$, then Alice performs a unitary transformation $U_{{0},{7+3+4-6},{0}}$ on the particle $Q_3$, then the GHZ state 
$\mid\varPsi_{5,6,1}\rangle$ will become 
$\mid\varPsi_{6,5,0}\rangle$. Last, Alice sends the particle $Q_2$ and $Q_3$ to the participant $P_1^{(1)}$. (this step also requires the use of decoy particles for security checks)

\indent\textbf{4.3.4} Participant $P_4^{(1)}$ sends the particle $Q_1$ to the participant $P_1^{(1)}$ at the same. After participant $P_1^{(1)}$ receives the particle $Q_1$, participant $P_1^{(1)}$ performs a unitary transformation
 $U_{1-5,4-5,0}=\sum_{j=0}^{6} \omega^{3j}\mid{j}\rangle\langle{j+6}\mid\notag$ on particle $Q_1$, then the GHZ state $\mid\varPsi_{6,5,0}\rangle$ will become 
$\mid\varPsi_{2,4,6}\rangle$,  $(\lambda_{A_1}^{(1)}s_1^{(1)}=3\times5=1(\bmod{7}), \lambda_{A_1}^{(1)}t_1^{(1)}=3\times6=4(\bmod{7}))$, where $\lambda_{A_1}^{(1)}$ is the private value of participant $P_1^{(1)}$, $s_1^{(1)}$ and $t_1^{(1)}$ are the corresponding 1 th component of the vectors $S^{(1)},T^{(1)}$. Eventually, participant $P_1^{(1)}$ performs GHZ-state joint measurement on the state $\mid\varPsi_{2,4,6}\rangle$, where $\mid\varPsi_{2,4,6}\rangle$ is new state obtained by Alice and 4 participants performing unitary transformation.

\indent\textbf{4.3.5} According to the GHZ-state joint measurement of $\mid\varPsi_{2,4,6}\rangle$, the participant $P_1^{(1)}$ can get the following three equations
\begin{align}
 \ s_1^\prime &=2,\\
 \ s_1^\prime &=4,\\
 \ s_1^\prime &=6,
\end{align}
The participant $P_1^{(1)}$ can respectively verify whether $H(s_1^\prime),H(s_2^\prime),H(s_3^\prime)$ and $H(2),H(4),H(6)$ are equal according to equations (8), (9) and (10). If the equation holds, participant $P_1^{(1)}$ can decide that the recovered secrets are correct; otherwise, it can judge that in the process of the secrets recovery, some participants are dishonest, the process can be stopped.\\
\indent Other participants $P_j^{(1)}$ in the authorized set $A_1$ can obtain secrets in the same way.

\section{Scheme analysis}\label{sec5}

\subsection{Correctness Analysis}\label{subsec1}

\begin{theorem}
 In this scheme, if a three-qudit GHZ state is $\mid\varPsi_{{u_1},{u_2},{u_3}}\rangle$, and a unitary transformation is $U_{{\alpha},{\beta},{0}}=\sum_{j=0}^{d-1} \omega^{j{\alpha}}\mid{j}\rangle\langle{j+\beta}\mid$, then, we can get
\begin{align}
    (U_{{\alpha},{\beta},{0}}\otimes I \otimes I)\mid\varPsi_{{u_1},{u_2},{u_3}}\rangle &=\mid\varPsi_{{u_1+\alpha},{u_2+\beta},{u_3+\beta}}\rangle \\
    (I \otimes U_{{\alpha},{\beta},{0}} \otimes I)\mid\varPsi_{{u_1},{u_2},{u_3}}\rangle &=\mid\varPsi_{{u_1+\alpha},{u_2-\beta},{u_3}}\rangle \\
    (I \otimes I \otimes U_{{\alpha},{\beta},{0}})\mid\varPsi_{{u_1},{u_2},{u_3}}\rangle &=\mid\varPsi_{{u_1+\alpha},{u_2},{u_3-\beta}}\rangle
\end{align} 
where $\omega=e^\frac{2\pi i}{d},u_i\in\lbrace0,1,\dots,d-1\rbrace, i=1,2,\dots,n$. The symbol "+" is an addition in the finite field $F_d$. 
\end{theorem}
\begin{proof}
    We konw that
\begin{gather*}
U_{{\alpha},{\beta},{0}}=\sum_{j=0}^{d-1} \omega^{j{\alpha}}\mid{j}\rangle\langle{j+\beta}\mid, \\
\mid\varPsi_{{u_1},{u_2},{u_3}}\rangle=\frac{1}{\sqrt{d}}\sum_{j=0}^{d-1} \omega^{j{u_1}}\mid{j,j+{u_2,j+u_3}}\rangle.
\end{gather*}
When performing $(U_{{\alpha},{\beta},{0}}$ on the first particle of $\mid\varPsi_{{u_1},{u_2},{u_3}}\rangle$, we can get 
\begin{align}
&\mathrel{\phantom{=}}
    (U_{{\alpha},{\beta},{0}}\otimes I \otimes I)\mid\varPsi_{{u_1},{u_2},{u_3}}\rangle \notag \\
&=(\sum_{j=0}^{d-1} \omega^{j_1{\alpha}}\mid{j_1}\rangle\langle{j+\beta}\mid\otimes I \otimes I)(\frac{1}{\sqrt{d}}\sum_{j_2=0}^{d-1} \omega^{{j_2}{u_1}}\mid{j_2,j_2+u_2,j_2+u_3}\rangle) \notag \\
&=\frac{1}{\sqrt{d}}\sum_{j_1=0}^{d-1}\sum_{j_2=0}^{d-1}\omega^{{j_1}\alpha+{j_2}{u_1}}\mid{j_1}\rangle\langle{j_1+\beta}\mid{j_2}\rangle\mid{j_2+u_2}\rangle\mid{j_2+u_3}\rangle \notag \\
&=\frac{1}{\sqrt{d}}\sum_{j_1=0}^{d-1}\omega^{{j_1}\alpha+({j_1+\beta}{u_1})}\mid{j_1}\rangle\mid{j_1+\beta+u_2}\rangle\mid{j_1+\beta+u_3}\rangle (j_2=j_1+\beta (\bmod{d}))\notag \\
&=\frac{1}{\sqrt{d}} \omega^{{\beta_1}{u_1}}\sum_{j_1=0}^{d-1}\omega^{{j_1}(\alpha+{u_1})}\mid{j_1}\rangle\mid{j_1+\beta+u_2}\rangle\mid{j_1+\beta+u_3}\rangle. \notag
\end{align} 
Ignoring the global, eventually, we have
\begin{equation}
U_{{\alpha},{\beta},{0}}\otimes I \otimes I) =\mid\varPsi_{{u_1+\alpha},{u_2+\beta},{u_3+\beta}}\rangle, \notag
\end{equation}
Similary, performing $U_{{\alpha},{\beta},{0}}$ on the second particle of $\mid\varPsi_{{u_1},{u_2},{u_3}}\rangle$ and performing $U_{{\alpha},{\beta},{0}}$ on the third particle of $\mid\varPsi_{{u_1},{u_2},{u_3}}\rangle$, we can get
\begin{align}
    (I \otimes U_{{\alpha},{\beta},{0}} \otimes I)\mid\varPsi_{{u_1},{u_2},{u_3}}\rangle &=\mid\varPsi_{{u_1+\alpha},{u_2-\beta},{u_3}}\rangle \notag\\
    (I \otimes I \otimes U_{{\alpha},{\beta},{0}})\mid\varPsi_{{u_1},{u_2},{u_3}}\rangle &=\mid\varPsi_{{u_1+\alpha},{u_2},{u_3-\beta}}\rangle. \notag
\end{align} 
\end{proof}

\subsection{Verifiability Analysis}\label{subsec2}

\begin{theorem}
 In this scheme, according to the GHZ-state joint measurement on the quantum state $\mid\varPsi_{{u_1^\prime},{u_2^\prime},{u_3^\prime}}\rangle$, the secrets can be obtained $s_i=u_i^\prime$, and the verification information can be obtained $H(s_i)=H(u_i^\prime)$, $i=1,2,3$. 
\end{theorem}
\begin{proof}
    In the secret recovery phase, according to Theorem 1, the quantum state $ \mid\varPsi_{\sum_{j=1}^{m}{\lambda_{A_i}^{(j)}s_j^{(i)}},\sum_{j=1}^{m}{\lambda_{A_i}^{(j)}t_j^{(i)}},\sum_{j=1}^{m}{\lambda_{A_i}^{(j)}t_j^{(i)}-s_2+s_3}}\rangle$ will be obtained after Alice and m participants performing the unitary transformation on the three particles of the quantum state $\mid\varPsi_{{u_1},{u_2},{u_3}}\rangle$, respectively. From the monotone span program, we can see that $\sum_{j=1}^{m}{\lambda_{A_i}^{(j)}s_j^{(i)}}=s_1,\sum_{j=1}^{m}{\lambda_{A_i}^{(j)}t_j^{(i)}}=s_2$, so the above quantum state is $\mid\varPsi_{{s_1},{s_2},{s_3}}\rangle$. Therefore, if a new quantum state $\mid\varPsi_{{u_1^\prime},{u_2^\prime},{u_3^\prime}}\rangle$ is performed on the GHZ-state joint measurement, then $s_1=u_1^\prime, s_2=u_2^\prime$ and $s_3=u_3^\prime$.\\
\indent The participant who finally recovers the secrets can verify the correctness of the recovered recovered secrets by comparing the value of the obtained secrets with the hash value published by Alice. Because of the unidirectivity and acollision resistance of the hash function, if the new quantum state $\mid\varPsi_{{u_1^\prime},{u_2^\prime},{u_3^\prime}}\rangle$ still satisfies the equation $H(u_1^\prime)=H(s_1), H(u_2^\prime)=H(s_2), H(u_3^\prime)=H(s_3)$, dishonest participants can send false information. Therefore, our scheme is verifiable.
\end{proof}

\subsection{Safety Analysis}\label{subsec3}

\subsubsection{Security Performance Of Access Structure}\label{subsubsec1}

\begin{theorem}
 In this scheme, a random authorized set in the access structure $\varGamma$ can recover all secrets; on the other hand, a random unauthorized set in the adversary sets $\Delta$ cannot recover all secrets. 
\end{theorem}
\begin{proof}
    The Definition of the access structure $\varGamma$ is : $ B\in\varGamma$, when $A\subseteq B\subseteq P $.  $\varGamma $ is the authorized set. Unauthorized sets are called adversary sets, denoted as $\Delta $, that is to say, $\Delta=\varGamma^c $.\\
\indent According to the definition of the access structure $\varGamma$, we can obtain $\varGamma \cap \Delta=\emptyset$. Therefore, in the selection process of the authorized set, it will never appear the condition of $\Delta \subseteq \varGamma$. In addition, Alice only sends the secret share to each participant in the authorized set, according to the monotone span program and linear secret sharing, all participants in the authorized set can obtain all secret shares, the participants perform unitary transformation to recover the secrets during the process of the particle transmission, the last participant can obtain the final secrets through GHZ-state joint measurement. Therefore, we can know a random unauthorized set, which cannot obtain all the secret shares needed to recover the secrets, so it is impossible to recover the final secrets.
\end{proof} 

\subsubsection{Attack Methods}\label{subsubsec2}

\bmhead{(i) Intercept-resend Attack}  Assume there exists an eavesdropper is Eve, she intercepts the particles sent by Alice or some participants. We all know the particles in the process of transmission are composed of some useful particles and decoy particles. Each decoy particle is randomly selected from the $Z$ basis or the $X$ basis, Alice and all participants randomly insert useful particles into the decoy particles. The eavesdropper Eve does not know the location of the insertion, the measurement basis, and the specific state of some decoy particles. Therefore, it can be detected during the security check between Alice and particles and between participants and participants. Moreover, the first participant uses a random number, even if the eavesdropper Eve gets useful particles, she cannot get real secret information.

\bmhead {(ii) Entanglement-measurement Attack}  In the process of particle transmission, the eavesdropper Eve uses the unitary transformation $U_E$ to entangle an auxiliary particle are $ \mid E \rangle$. If the decoy particle is randomly selected from the $Z$ basis, then the unitary transformation $U_E$ first acts on the decoy particle to obtain:  
\begin{equation}
U_E\mid i\rangle\mid E \rangle =\sum_{j=1}^{m}a_{ij}\mid j \rangle\mid e_{ij}\rangle, \notag
\end{equation}
where $\mid e_{ij}\rangle (i,j\in\lbrace{0,1,\dots,d-1}\rbrace)$ are the states determined by the unitary operation $U_E$, and $\sum_{j=1}^{d-1}{\mid a_{ij}\mid}^2=1(i=0,1,\dots,d-1)$.\\
\indent In order to avoid the eavesdropping check, Eve has to set : $a_{ij}=0$, where $i\ne j$, and $i,j \in\lbrace{0,1,\dots,d-1}\rbrace$. Thus, the influence of $U_E$ performed on the decoy particle is simplified as
\begin{equation*}
\ U_E\mid i\rangle\mid E\rangle=a_{ii}\mid i\rangle\mid e_{ii}\rangle, i=0,1,\dots,d-1.
\end{equation*}
\indent If the decoy particle is in the $X$ basis, the influence of $U_E$ performed on the decoy particle is as follows :
\begin{equation*}  \begin{split}
\ U_E\mid J_j\rangle\mid E\rangle &=\ U_E(\frac{1}{\sqrt{d}}\sum_{j_1=0}^{d-1}\mid k\rangle)\mid E\rangle \\
 &=\frac{1}{\sqrt{d}}\sum_{j_1=0}^{d-1}\omega^{kj}U_E\mid k\rangle\mid E\rangle \\
 &=\frac{1}{\sqrt{d}}\sum_{j_1=0}^{d-1}\omega^{kj}a_{kk}\mid k\rangle\mid e_{kk}\rangle
\end{split}   \end{equation*}
where $j \in \lbrace{0,1,\dots,d-1}\rbrace$.\\
\indent In addition, we all know $\mid j\rangle=\frac{1}{\sqrt{d}}\sum_{j_1=0}^{d-1}\omega^{-kj}\mid J_k\rangle$, so
\begin{equation*}
\ U_E\mid J_j\rangle\mid E\rangle=\frac{1}{d}\sum_{k=0}^{d-1}\sum_{m=0}^{d-1}\omega^{k(j-m)}a_{kk}\mid J_m\rangle\mid e_{kk}\rangle.
\end{equation*}
\indent In order to avoid the eavesdropping check, Eve has to set
\begin{equation*}
\sum_{m=0}^{d-1}\omega^{k(j-m)}a_{kk}\mid e_{kk}\rangle=0.
\end{equation*}
where $m \in \lbrace{0,1,\dots,d-1}\rbrace$ and $m\ne j$. Then for any $j\in \lbrace{0,1,\dots,d-1}\rbrace$, we can get d-1 equations.
\begin{gather*}
    a_{00}\mid e_{00}\rangle+\omega a_{11}\mid e_{11}\rangle+\cdots+\omega^{d-1} a_{(d-1)(d-1)}\mid e_{(d-1)(d-1)}\rangle=0, \\
    a_{00}\mid e_{00}\rangle+\omega^2 a_{11}\mid e_{11}\rangle+\cdots+\omega^{2(d-1)} a_{(d-1)(d-1)}\mid e_{(d-1)(d-1)}\rangle=0, \\
    \vdots \\
    a_{00}\mid e_{00}\rangle+\omega^{d-1} a_{11}\mid e_{11}\rangle+\cdots+\omega^{{d-1}^2} a_{(d-1)(d-1)}\mid e_{(d-1)(d-1)}\rangle=0. \\
\end{gather*}
Based on these d-1 equations, we can compute that 
\begin{equation*}
a_{00}\mid e_{00}\rangle=a_{11}\mid e_{11}\rangle=\cdots=a_{(d-1)(d-1)}\mid e_{(d-1)(d-1)}\rangle.
\end{equation*}
Then, we obtain
\begin{equation*}  \begin{split}
\ U_E\mid J_j\rangle\mid E\rangle &=\frac{1}{d}\sum_{k=0}^{d-1}\sum_{m=0}^{d-1}\omega^{k(j-m)}\mid J_m\rangle\otimes a_{00}\mid e_{00}\rangle \\
 &=\frac{1}{d}\sum_{k=0}^{d-1}\omega^{kj}\mid k\rangle\otimes a_{00}\mid e_{00}\rangle.
\end{split}   \end{equation*}
\indent Therefore, Eve will get the same information from the ancillary particle regardless of the state of the useful particle, and cannot eavesdrop on the secret information. In conclusion, the attack of entangle-measurement fails.

\bmhead{(iii) Participant Attack}  In the quantum secret sharing scheme, the influence of dishonest participants on the scheme is destructive. In our scheme, we use decoy particles randomly selected from the $Z$ basis and the $X$ to ensure the security transmission of useful particles. In addition, the participants do not know the position, state and correct measurement basis of the decoy particles during the process of transmitting the particles. Therefore, if the participants launch attacks, it will be found. On the other hand, from the above analysis, we can also know that due to the existence of decoy particles, even if dishonest participants use intercept-resend attacks, they cannot obtain true secret information. Therefore, our scheme can resist participant attack.

\bmhead{(iv) Collusion Attack}  Collusion attack is a more serious attack than participant attack. In the worst case, we assume that only Alice and the participant who finally recovered the secret are honest, and that the remaining $m-1$ participant wants to unite to recover the secret. They can only get $\sum_{j=2}^{m}\lambda_{A_i}^{(j)}s_j^{(i)}$ 
and $\sum_{j=2}^{m}\lambda_{A_i}^{(j)}t_j^{(i)}$ by exchanging their secret shares. However, the share information of Alice and the participant who finally recovered the secret could not be obtained. On the other hand, the participant who finally recovers the secret initially embeds a random number in the quantum state. Even if $m-1$ participants perform unitary transformation on the quantum state and makes measurement, the correct measurement information cannot be obtained, so the secret cannot be recovered.

\section{Efficiency Comparison}\label{sec6}

Here we will compare the proposed quantum secret sharing scheme with the existing four quantum secret sharing schemes in following four aspects : the number of measured quantum bits, the number of unitary transformations, the number of the recovered secrets and the security of these schemes. The existing four quantum secret sharing schemes are Yang \citep{ref-ur15}, Qin \citep{ref-ur16}, Mashhadi \citep{ref-ur20} and Wu \citep{ref-ur21}.\\
\indent(1) Yang \citep{ref-ur15} scheme   It proposes a $(n,n)$ threshold multidimensional quantum secret sharing scheme based on QFT, in which only when $n$ participants want to rebuild the secret together, and if $n$ participants want to rebuild the secret together, they must share an $n$-particle entangled state, the total number of transmitted message particles is $n-1$, and the number of final recovered secrets is $1$. In addition, all participants need to make $n$ measurements and perform $n$ QFT and $n$ unitary operations respective. Since this scheme does not insert decoy particles, it cannot resist eavesdropping attacks. \\
\indent(2) Qin \citep{ref-ur16} scheme   It proposes a verifiable $(t,n)$ threshold quantum secret sharing scheme by using d-dimensional Bell state and Lagrange interpolation. If $t$ participants want to rebuild the secret together, Alice needs to prepare a d-dimensional Bell state, each participant needs to prepare $L$ decoy particles, the total number of transmitted message particles is $m-1$, and the number of final recovered secrets is $1$. In addition, all participants need to make 1 measurements and perform $m$ unitary operations. \\
\indent(3) Mashhadi \citep{ref-ur20} scheme   It designs a hybrid quantum secret sharing scheme based on QFT and monotone span program. If an authorized set(including $m$ participants) wants to reconstruct the secret, each participant needs to prepare $L$ decoy particles, the total number of transmitted message particles is $m-1$, and the number of final recovered secrets is $1$. In addition, all participants need to make $m$ measurements and perform $m+1$ QFT, $m-1$ SUM and $m$ order-unitary operations respective. This scheme uses twice hash function to verify the security of the scheme.  \\
\indent(4) Wu \citep{ref-ur21} scheme   It proposes a general multidimensional quantum secret sharing scheme, in which only the participants in the authorized set can reconstruct the secret. If an authorized set(containing $m$ participants) wants to reconstruct the secret, Alice needs to prepare a three-qudit GHZ state, each participant needs to prepare $L$ decoy particles, the total number of transmitted message particles is $m-1$, and the number of final recovered secrets is $1$. In addition, all participants need to make one measurement and perform $m+2$ unitary operations respective. This scheme uses two hash function to check the security of scheme.  \\
\indent(5) Our scheme designs a efficient quantum secret sharing scheme based on the access structure and monotone span program, in which only the participants in the authorized set can reconstruct the secret. If an authorized set(containing $m$ participants) wants to reconstruct the secret, Alice needs to prepare a three-qudit GHZ state, each participant needs to prepare $L$ decoy particles, the total number of transmitted message particles is $m-1$, and the number of final recovered secrets is $3$. In addition, all participants need to make one measurement and perform $m+3$ unitary operations respective. 
\indent Table 1 shows a detailed comparison between these schemes and our proposed schemes. The analysis shows that our schemes not only ensure security, but also avoid the waste of quantum resources.

\begin{sidewaystable}
\sidewaystablefn%
\begin{center}
\begin{minipage}{\textheight}
\caption{Comparison between our scheme and literature \citep{ref-ur15}, \citep{ref-ur16}, \citep{ref-ur20}, \citep{ref-ur21}}\label{tab1}
\begin{tabular*}{\textheight}{@{\extracolsep{\fill}}lcccccc@{\extracolsep{\fill}}}
\toprule%
 & Yang \citep{ref-ur15}  & Qin \citep{ref-ur16} & Mashhadi \citep{ref-ur20} & Wu \citep{ref-ur21} & Ours\\
\midrule
Model    & $(n,n)$   & $(t,n)$  & General & General & General \\
Numbers of message particles  & $n-1$   & $t-1$  & $m-1$ & $m$& $m$ \\
Numbers of unitary operation  & $n$   & $t$  & $m$ & $m+2$& $m+3$ \\
Numbers of measurement operation  & $n$   & $1$  & $1$ & $1$ & $1$ \\
Numbers of recovered secrets  & 1   & 1  & 1 & 1 & 3 \\
Eavesdropping check  & -   & Decoy particles  & Decoy particles & Decoy particles & Decoy particles \\
Entanglement measurement  & -   & Yes  & Yes  & Yes & Yes\\
Verification of secret  & -   & Once hash  & Twice hash  & Twice hash & Once hash\\
\botrule
\end{tabular*}
\end{minipage}
\end{center}
\end{sidewaystable}


\section{Conclusion}\label{sec7}

In this paper, we propose a efficient multiple secret sharing scheme based on monotone span program using GHZ state and linear secret sharing on a general access structure. In this scheme, the distributor can transfer GHZ state among participants to share multiple secrets. Participants perform appropriate unitary transformation on the quantum state according to their secret share, and participants who want to obtain the secret perform a GHZ-state joint measurement on the final quantum state. The scheme uses decoy particles for eavesdropping check, and uses hash function to check the correctness of the recovered secrets. It has less verification information, avoids the waste of quantum state, and achieves higher security. In addition, the security analysis shows that the proposed scheme can resist interception-retransmission attacks, entanglement-measurement attacks, and participant attacks. The scheme can transmit more information by extending quantum bits, and will have more applications in the future.

\backmatter

\section*{Declarations}

\begin{itemize}
\item Ethics Approval and Consent to participate : Not applicable.
\item Consent for publication : All authors have read and agreed to the published version of the manuscript.
\item Availability of supporting data : Data sharing not applicable to this article as no datasets were generated or analysed during the current study.
\item Competing interests : The authors declare no conflict of interest. 
\item Funding : This research was funded by National Natural Science Foundation of China, grant number 11671244.
\item Authors' contributions : Conceptualization, Z.L. and S.L.; writing-original draft preparation, S.L.; writing-review, Z.L. and J.G.; editing, S.L. and D.M.
\end{itemize}

\noindent


\begin{appendices}




\end{appendices}




\end{document}